\documentclass[final, 3p, times]{elsarticle}

\usepackage{latexsym, amsmath, amssymb,amsthm, graphicx}
\usepackage[bookmarks, a4paper, colorlinks=true]{hyperref}

\newtheorem{theorem}{Theorem}
\newtheorem{lemma}{Lemma}

\newtheorem{remark}{Remark}

\begin{document}

\begin{frontmatter}

\title{Solution Representations of Solving Problems for the Black-Scholes equations and Application to the Pricing Options on Bond with Credit Risk}
\author{Hyong-Chol O ${}^{*}$, Dae-Song Kim, Dae-Song  Choe}
\ead{hc.o@ryongnamsan.edu.kp}
\address{Faculty of Mathematics, \textbf{Kim Il Sung} University, Pyongyang, Democratic People's Republic of Korea}

\begin{abstract}
In this paper is investigated the pricing problem of options on bonds with credit risk based on analysis on two kinds of solving problems for the Black-Scholes equations. First, a solution representation of the Black-Scholes equation with the maturity payoff function which is the product of the power function, normal distribution function and characteristic function is provided. Then a solution representation of a special terminal boundary value problem of the Black-Sholes equation is provided and its monotonicity is proved. The simplest case of the structural model of the credit bond is studied and its pricing formula is provided, and based on the results, the pricing model of option on corporate bond with credit risk is transformed into a terminal boundary value problem of the Black-Scholes equation with some special maturity payoff functions and the solution formula is obtained. Using it, we provide the pricing formulae of the puttable and callable bonds with credit risk.
\end{abstract}


\end{frontmatter}

\section{Introduction}
Solving problems for the Black-Scholes equation become pricing models of various financial derivatives according to the forms of maturity payoff functions, and it is important to obtain solution representations with complicated maturity payoff functions. \cite{Bu2004} provides solution representations of the terminal value problems of the Black-Scholes equation with the maturity payoff functions which are the products of power functions with  exponents 0 and 1, cumulative distribution functions of normal distribution and characteristic functions, and using them provides the pricing formulae of the exotic options with two expiry dates. \cite{OC2019} provides a solution representation of the terminal value problem of the Black-Scholes equation with the maturity payoff function which is the product of the power function with any real number exponent and cumulative distribution function of normal distribution, and using it obtains the pricing formulae of some the exotics including savings plans that provide a choice of indexing.

On the other hand, recent research has been carried out on pricing problems for complicated composite options such as bond options. \cite{{FNV2016}, {ZEW2017},  {ZW2012}, {ZYL}} studied the pricing problem of American bond options and \cite{{RF2016}, {RF2017}} studied the pricing coupon bond options as interest rate-derivatives under the Hull-White model. These are all the pricing problems of bond options without credit risk. \cite{OK2020} provides a partial differential equation model of the price of a Discrete Coupon Bond with Early redemption provision (with credit risk) and an analytical solution representation, their approach is to analyze together by adding an early-redemption provision to the pricing model of the discrete-coupon bonds of \cite{OJKJ2017} rather than from the viewpoint of the bond option.

The purpose of this paper is to investigate the pricing problem of options on bonds with credit risk based on analysis on two kinds of solving problems for the Black-Scholes equations. According to experts, the pricing problem of options on bonds with credit risk is a challenging problem. All corporate bonds with an early redemption provision may be considered as a bond option with credit risk. 

In this paper, we study a simple structural model of pricing option on bond with credit risk. 

We first provide a solution representation of the Black-Scholes equation with the maturity payoff function which is the product of the power function, normal distribution function and characteristic function. This solution representation includes all the solution representations given in \cite{{Bu2004}, {OK2013}, {OC2019}} as special cases. Then we provide a solution representation of a special terminal boundary value problem of the Black-Sholes equation and prove the monotonicity of its solution. Unlike the monotonicity results of the solution to the terminal value problem of the Black-Scholes equation (e.g., \cite{OJK2016} et al), the monotonicity of this solution to the terminal boundary value problem of the Black-Scholes equation, is not easily derived and is proved based on a rigorous estimates of the cumulative distribution function of normal distribution and its derivatives.

Then we provide the simplest case of the structural model of the straight credit bond (i.e. credit bond without early redemption provision) and its pricing formula and prove the monotonicity of its pricing function. Then using it, we transform the pricing model of option on corporate bond with credit risk to a terminal boundary value problem of the Black-Scholes equation with some special maturity payoff functions and obtain a solution formula. Using it, we provide the pricing formulae of the puttable and callable bonds with credit risk.

The rest of the paper is organized as follows. In Section 2, we provide  solution representations of two kinds of solving problems for the Black-Scholes equations arising from investigation of pricing options on credit bonds and study the monotonicity of the solutions. In Section 3 we provide a pricing formula of the simplest case of the structural model for pricing credit bond (a terminal boundary value problem of the Black-Scholes equation) and the pricing formulae of puttable and callable bonds with credit risk.

\section{Solving problems for the Black-Scholes equations and the monotonicity of the solutions}

\subsection{A solution representation of a terminal value problem of Black-Scholes equation with a special maturity payoff functions}

The following notation is used:

\[ \delta \left(\frac{X^{i} }{L} ,\; \tau _{1} ,\; \tau _{2} ,\; \tau _{3} \right)=\left[\ln \frac{X^{i} }{L} +\left(r-q-\frac{\sigma ^{2} }{2} \right)\tau _{1} +\sigma ^{2} \tau _{2} \right](\sigma \sqrt{\tau _{3} } )^{-1} (i \in \mathbb{R}).\]

\begin{theorem} For the terminal value problem
\begin{equation} \label{eq1} 
\frac{\partial V}{\partial t} +\frac{1}{2} \sigma ^{2} X^{2} \frac{\partial ^{2} V}{\partial X^{2} } +(r-q)X\frac{\partial V}{\partial X} -rV=0,\; \; 0<t<T,\; \; X>0 
\end{equation} 
\begin{equation} \label{eq2}
V(X,\; T)=X^{\beta } N[\delta (X^{i} /L,\; \; \tau _{1} ,\; \; \tau _{2} ,\; \; \tau _{3} )]\cdot 1\{X>K\}  
\end{equation} 
we have the following solution representation
\begin{equation} \label{eq3}
V(X,\; \tau ;\; \tau _{1} ,\; \; \tau _{2} ,\; \; \tau _{3} ,\; \; i,\; \; L,\; \; K)=X^{\beta } e^{\mu (\beta )\tau } N_{2} (a_{1} ,\; a_{2} ;\; \rho).               
\end{equation}
Here
\[N(x)=\frac{1}{\sqrt{2\pi } } \int _{-\infty }^{x}e^{-t^{2} /2} dt , \mu (\beta )=-r+\beta \left[r-q-\frac{\sigma ^{2} }{2} (1-\beta )\right], \]
\[\tau =T-t, \;\; \rho =i\sqrt{\frac{\tau }{\tau _{3} +i^{2} \tau } } ,\]
\[N_{2} (a_{1} ,\; a_{2} ;\; \rho)=\frac{1}{2\pi \sqrt{1-\rho ^{2}} } \int _{-\infty }^{\; a_{1} }\int _{-\infty }^{\; \; a_{2} }e^{-\frac{y_{1} {}^{2} -2\rho y_{1} y_{2} +y_{2} {}^{2} }{2(1-\rho ^{2} )} }  dy_{2} dy_{1}  \]
\[a_{2} =\pm \left[\ln \frac{X}{K} +\left(r-q-\frac{\sigma ^{2} }{2} \right)\tau +\sigma ^{2} \beta \tau \right](\sigma \sqrt{\tau } )^{-1} =\pm \delta (X/K,\; \; \tau ,\; \; \beta \tau ,\; \; \tau )\] 
\[\begin{array}{l} {a_{1} =\left[\ln \frac{X^{i} }{L} +\left(r-q-\frac{\sigma ^{2} }{2} \right)(\tau _{1} +i\tau )+\sigma ^{2} (\tau _{2} +i\beta \tau )\right](\sigma \sqrt{\tau _{3} +i^{2} \tau } )^{-1} =}\;\;\;\; \\ {\; \; \; \; =\delta (X^{i} /L,\; \; \tau _{1} +i\tau ,\; \; \tau _{2} +i\beta \tau ,\; \; \tau _{3} +i^{2} \tau )\; .} \end{array}\]
\end{theorem}
\begin{proof} By the formula (3) of \cite{OK2013}, the solution of \eqref{eq1} and \eqref{eq2} is written as follows:
\[V(X,\; \; \tau)=\frac{e^{-r\tau } }{\sigma \sqrt{2\pi \tau } } \int _{0}^{+\infty }\frac{z^{\beta } }{z} \left[\frac{1}{\sqrt{2\pi } } \int _{-\infty }^{\delta (z^{i} /L)}e^{-\frac{y^{2} }{2} } dy \right]\cdot 1\{ z\ge K\} \cdot e^{-\frac{1}{2\sigma ^{2} \tau } \left[\ln \frac{X}{z} +(r-q-\sigma ^{2} /2)\tau \right]^2} dz \] 
\begin{equation} \label{eq4} 
=\frac{e^{-r\tau } }{\sigma \sqrt{2\pi \tau } } \int _{K}^{+\infty }\frac{z^{\beta } }{z} \left[\frac{1}{\sqrt{2\pi } } \int _{-\infty }^{\delta (z^{i} /L)}e^{-\frac{y^{2} }{2} } dy \right]\cdot e^{-\frac{1}{2\sigma ^{2} \tau } \left[\ln \frac{X}{z} +(r-q-\sigma ^{2} /2)\tau \right]^2} dz
\end{equation} 
Here, noting that $\exp (\beta \ln z)=z^{\beta } $ and using the transformation
\[y_{2} =\left[\ln \frac{X}{z} +\left(r-q-\frac{\sigma ^{2} }{2} +\beta \sigma ^{2} \right)\tau \right](\sigma \sqrt{\tau } )^{-1} \] 
\[y_{1} =\frac{y\sqrt{\tau _{3} } }{\sqrt{\tau _{3} +i^{2} \tau } } +i\left[\ln \frac{X}{z} +\left(r-q-\frac{\sigma ^{2} }{2} +\beta \sigma ^{2} \right)\tau \right]\left(\sigma \sqrt{\tau _{3} +i^{2} \tau } \right)^{-1},\] 
then we have
\[z>K\Leftrightarrow y_{2} <\left[\ln \frac{X}{K} +\left(r-q-\frac{\sigma ^{2} }{2} +\beta \sigma ^{2} \right)\tau \right]\left(\sigma \sqrt{\tau } \right)^{-1} =a_{2} ,\] 
and noting that
\[ y=y_{1} \sqrt{\frac{\tau _{3} +i^{2} \tau }{\tau _{3} } } -iy_{2} \sqrt{\frac{\tau }{\tau _{3} } } ,\]
then \eqref{eq4} can be written as follows:
\begin{equation} \label{eq5} 
V(X,\; \tau ;\; \tau _{1} )=\frac{X^{\beta } e^{\mu (\beta )\tau } }{2\pi } \sqrt{\frac{\tau _{3} +i^{2} \tau }{\tau _{3} } } \int _{-\infty }^{a_{1} }\int _{-\infty }^{a_{2} }e^{-\frac{\left(\sqrt{\frac{\tau _{3} +i^{2} \tau }{\tau _{3} } } y_{1} -i\sqrt{\frac{\tau }{\tau _{3} } } y_{2} \right)^{2} +y_{2} {}^{2} }{2} }  dy_{2} dy_{1}  .      
\end{equation} 
Now let $\rho =i\sqrt{\tau /(\tau _{3} +i^{2} \tau )} $, then \eqref{eq5} can be written as \eqref{eq3} using two dimensionalnormal distribution function with the correlation matrix $\left(\begin{array}{cc} {1} & {\rho } \\ {\rho } & {1} \end{array}\right)$.
\end{proof}

\begin{remark}
The formula \eqref{eq3} in the case that $\beta =1$or 0, $i=1,\; \; \tau _{1} =\tau _{3} =T-t,$ $\tau _{2} =\beta (T-t)$ is just the pricing formula of the second order binary option of \cite{Bu2004}, the formula \eqref{eq3} in the case that $\beta $ is any real number and $i=1,\; \; \tau _{1} =\tau _{3} =$ $T-t,$ $\tau _{2} =\beta (T-t)$ is just the pricing formula of the second order power binary option (the expression (40) at page 21 of \cite{OC2019}) and the formula \eqref{eq3} with $K=0$ is the just the pricing formula of the power normal distribution standard option (the expression (12) at page 5 of \cite{OC2019}).
\end{remark}

\subsection{A solution representation of a terminal boundary value problem of the Black-Scholes equation and the monotonicity of its solution}

Consider the following problem:
\begin{equation} \label{eq6} 
\frac{\partial w}{\partial t} +\frac{1}{2} \sigma ^{2} x^{2} \frac{\partial ^{2} w}{\partial x^{2} } +cx\frac{\partial w}{\partial x} =0,\; \; 1<x<+\infty ,\; \; 0<t<T 
\end{equation} 
\begin{equation} \label{eq7} 
w(1,\; \; t)=0,\; \; 0\le t<T 
\end{equation} 
\begin{equation} \label{eq8} 
w(x,\; T)=1,\; \; x>1 
\end{equation} 
This is a special terminal boundary value problem of the Black-Scholes equation and the pricing problem of barrier option with zero interest rate, dividend rate $-c$ and the maturity payoff 1. Thus it can be solved by the image solution method \cite{Bu2001}. First, appling the formula \eqref{eq3} with $\beta =0,\; \; i=0,\; \; L=0,\; \; K=1,\; \; r=0,\; \; q=-c$ under consideration $N_{2} (+\infty ,\; a_{2} ;\; A)=N(a_{2} )$, we have the solution of \eqref{eq6} and \eqref{eq8}:
\[u(x,\; \; t)=N\left(\frac{\ln x+(c-\sigma ^{2} /2)(T-t)}{\sigma \sqrt{T-t} } \right).\] 
Then using the image solution method, the solution of \eqref{eq6}, \eqref{eq7}, \eqref{eq8} is provided by
\[w(x,\; \; t)=u(x,\; \; t)-x^{1-2c/\sigma ^{2} } u(1/x,\; \; t).\] 
That is
\begin{equation} \label{eq9} 
w(x,\; \; t)=N\left(\frac{\ln x+(c-\frac{\sigma ^{2} }{2} )(T-t)}{\sigma \sqrt{T-t} } \right)-x^{1-\frac{2c}{\sigma ^{2} } } N\left(\frac{\ln \frac{1}{x} +(c-\frac{\sigma ^{2} }{2} )(T-t)}{\sigma \sqrt{T-t} } \right).    
\end{equation}

\label{theorem2} \begin{theorem} The solution of \eqref{eq6}, \eqref{eq7}, \eqref{eq8} is provided by \eqref{eq9}. Furthermore, we have the following estimates:
\[w_{x} (x,\; \; t)>0,~0<w(x,\; \; t)<1(x>1,\; \; t<T).\] 
\end{theorem}
\noindent ${Proof.}$ Let $R=c\cdot (T-t),\; \; S=\sigma \sqrt{T-t} (>0)$. Then we have
\[w(x,\; \; t)=N\left(\frac{2R-S^{2} }{2S} +\frac{1}{S} \ln x\right)-x^{1-\frac{2R}{S^{2} } } N\left(\frac{2R-S^{2} }{2S} -\frac{1}{S} \ln x\right).\] 
Now, let $\frac{2R-S^{2} }{2S} =\xi $, then $w(x,\; \; t)$ is written as follows

\[f(x\; ;S,\; \; \xi )=N\left(\xi +\frac{1}{S} \ln x\right)-x^{-\frac{2\xi }{S} } N\left(\xi -\frac{1}{S} \ln x\right)\;\;\;(S>0,\; \; \xi \in R \text{: parameters})\].

\noindent Evidently, we have
\begin{equation} \label{eq10} 
f(0+)=f(1)=0,\; \; f(+\infty )=1.                        
\end{equation} 
And we have
\begin{equation} \label{eq11} 
f'(x)=\frac{2}{S} x^{-\frac{2\xi }{S} -1} \left[\frac{1}{\sqrt{2\pi } } \exp \left\{-\frac{1}{2} \left[\xi -\frac{1}{S} \ln x\right]^{2} \right\}+\xi {\kern 1pt} N\left(\xi -\frac{1}{S} \ln x\right)\right]=\frac{2}{S} x^{-\frac{2\xi }{S} -1} \cdot g(x).  
\end{equation} 
Thus if $\xi \ge 0$, then $f'(x)>0(\forall x>0)$. 

Now assume that $\xi <0$. Note that
\[g(x)=\frac{1}{\sqrt{2\pi } } \exp \left\{-\frac{1}{2} \left[\xi -\frac{1}{S} \ln x\right]^{2} \right\}+\xi {\kern 1pt} N\left(\xi -\frac{1}{S} \ln x\right).\] 
Then we have
\begin{equation} \label{eq12} 
g(0+)=\xi <0, g(1)=\frac{1}{\sqrt{2\pi } } e^{-\xi ^{2} /2} +\xi {\kern 1pt} N(\xi ),\; \; g(+\infty )=0.         
\end{equation} 
The proof of Theorem 2 needs the following Lemmas.

\label{lemma1} \begin{lemma} $f(x)<0,\; \; f'(x)<0$ in a right neighborhood of $x=0$. In particular, if $-S/2<\xi <0$, then $f'(0+)=-\infty $; if $\xi =-S/2$, then $f'(0+)=-1<0$ and if $\xi <-S$ $/2$, then $f'(0+)=0$.
\end{lemma}
\begin{proof} It is evident from \eqref{eq11} and the first expression of \eqref{eq12}, and the first expression of \eqref{eq10}. 
\end{proof}

\label{lemma2} \begin{lemma} $g(1)=\frac{1}{\sqrt{2\pi } } e^{-\xi ^{2} /2} +\xi {\kern 1pt} N(\xi )>0$.
\end{lemma}
\begin{proof} Since $g'(x)=-(S\sqrt{2\pi } \; x)^{-1} \exp \{ -[\xi -S^{-1} \ln x]^{2} /2\} \cdot \ln x$, $g(x)$ is strictly increasing for $0<x<1$ and strictly decreasing for $x>1$, and thus we have $g(1)=$ $\max g(x)$. From the last expression, we have the conclusion.
\end{proof}

Using Lemma 2, Theorem 2 is soon proved. In fact, $g(x)$ is strictly increasing in the interval $0<x<1$, $g(0+)<0$ from the first expression of \eqref{eq12} and $g(1)>0$ from Lemma 2, there exists a unique $x_{0} $ in the interval [0, 1] such that $g(x_{0} )=0$ and $g(x)\ne 0$ in the interval $[1,\; \; +\infty )$. Thus $x_{0} $ is the unique point such that $f'(x_{0} )=0$ and therefore $f'(x)<0$on $(0,\; \; x_{0} )$ and $f'(x)>0$ on $[x_{0} ,\; \; +\infty )$. So we have $w_{x} (x,\; \; t)>0$($x>1$) and $0=w(1,\; \; t)\le w(x,\; \; t)<w(+\infty ,\; \; t)=1$.

\section{Application to the pricing problem for option on bonds with credit risk. }

\subsection{One factor structural model for straight credit bond}

Main Assumptions are as follows:

\textbf{Assumption 1}: Short-term interest rate $r$ is constant. 

\textbf{Assumption 2}: The firm value $V$ consists of $m$ shares of traded assets $S$ and $n$ sheets of zero coupon bonds:
\[V_{t} =mS_{t} +nB_{t} \] 
and follows the geometric Brownian motion

\[dV=(r-q)Vdt+\sigma VdW \;\text{(under the risk neutral measure)}.\]

\noindent The firm pays out dividend in the ratio of $q$(a constant) for a unit of firm value, continuously (volatility $\sigma $: constant)

\textbf{Assumption 3}: Expected default occurs when

\[V\le V_{b} (t) \;\text{(Here}\; V_{b} (t)=be^{-a(T-t)} , a,\; b: \text{constants).}\]

\textbf{Assumption 4}: Default recovery is $R_{d} =R\cdot e^{-r(T-t)} $ ($0\le R<1$:constant).

\textbf{Assumption 5}: Our Corporate bond price is given by a sufficiently smooth deterministic function $B(V,\; \; t)$. The face value of the bond (the price at maturity $T$) is 1(unit of currency).

According to the argument in Section 3 of \cite{OW2013}, $B=B(V,\; t)$ is the solution to the following solving problem for the terminal vlue problem of Black-Scholes equation:
\begin{equation} \label{eq13} 
\frac{\partial B}{\partial t} +\frac{1}{2} \sigma ^{2} V^{2} \frac{\partial ^{2} B}{\partial V^{2} } +(r-q)V\frac{\partial B}{\partial V} -rB=0,\; \; V_{b} (t)<V<+\infty ,\; \; 0<t<T 
\end{equation} 
\begin{equation} \label{eq14} 
B(V_{b} (t),\; \; t)=e^{-r(T-t)} R,\; \; 0<t<T 
\end{equation} 
\begin{equation} \label{eq15} 
B(V,\; T)=1,\; \; V>V_{b} (T)=b 
\end{equation} 

\begin{theorem} The solution to the problem \eqref{eq13}$\mathrm{\sim}$\eqref{eq15} is provided as follows:
\begin{equation} \label{eq16} 
B(V,\; \; t)=R\cdot e^{-r(T-t)} +W(V,\; \; t)(1-R)e^{-r(T-t)}  
\end{equation} 
Here $W(V,t)$ is the surviving probability of the bond, that is, the probability of no default at time $t(0\le t<T)$ and given as follows:
\begin{equation} \label{eq17} 
W(V,\; \; t)=N(d_{1} )-\left(\frac{V}{be^{-a(T-t)} } \right)^{1-\frac{2(r-q-a)}{\sigma ^{2} } } N(d_{2} ) 
\end{equation} 
\begin{equation} \label{eq18} 
\begin{array}{l} {d_{1} =\frac{\ln \frac{V}{be^{-a(T-t)} } +(r-q-a-\frac{\sigma ^{2} }{2} )(T-t)}{\sigma \sqrt{T-t} } =\frac{\ln \frac{V}{b} +(r-q-\frac{\sigma ^{2} }{2} )(T-t)}{\sigma \sqrt{T-t} } } \\ {d_{2} =\frac{\ln \frac{be^{-a(T-t)} }{V} +(r-q-a-\frac{\sigma ^{2} }{2} )(T-t)}{\sigma \sqrt{T-t} } =\frac{\ln \frac{b}{V} +(r-q-2a-\frac{\sigma ^{2} }{2} )(T-t)}{\sigma \sqrt{T-t} } } \end{array} 
\end{equation} 
Furthermore, we have 
\[B_{V} (V,\; t)>0,\; \;R\cdot e^{-r(T-t)}<B(V,\; t)<e^{-r(T-t)}, \; \; V>V_{b} (t).\] 
\end{theorem}
\begin{proof} The problem  \eqref{eq13}$\mathrm{\sim}$\eqref{eq15} is a special case of the problem (3.11) of \cite{OW2013} and thus can be solved by the method of \cite{OW2013}. At time $0\le t<T$, the bond holder can receive at least the default recovery $R\cdot e^{-r(T-t)} $ regardless of default and in the case of no default at $t$ he can receive the remainder $(1-R)e^{-r(T-t)} $. Thus if $W(V,\; \; t)$ is the probability of no default at time $t$, then our bond price will be represented by \eqref{eq16}. Then $W(V,\; \; t)$ is the solution to the following problem:
\[\frac{\partial W}{\partial t} +\frac{1}{2} \sigma ^{2} V^{2} \frac{\partial ^{2} W}{\partial V^{2} } +(r-q)V\frac{\partial W}{\partial V} =0,\; \; V_{b} (t)<V<+\infty ,\; \; 0<t<T\] 
\[W(V_{b} (t),\; \; t)=0,\; \; 0\le t<T\] 
\[W(V,\; T)=1,\; \; V>b.\] 
If we use the new variables$V/(be^{-a(T-t)} )=x,\; \; W(V,\; \; t)=w(x,\; \; t)$, then this problem is transformed to the problem \eqref{eq6}, \eqref{eq7}, \eqref{eq8} with $c=r-q-a$ and thus $w(x,\; \; t)$ is provided by \eqref{eq9}. If we return to the original variables, then we obtain \eqref{eq17}, \eqref{eq18}. 

From Theorem 2, we have $0<W(V,\; \; t)<1$, $W_{V} (V,\; \; t)>0$ and form this result we have $B_{V} (V,\; t)>0,\; \;R\cdot e^{-r(T-t)}<B(V,\; t)<e^{-r(T-t)}$.
\end{proof}

\subsection{A Pricing Model for Bond Option and Pricing Formula.}

Now we assume that the bond issuing company added to the bond defined by the Assumptions 1$\mathrm{\sim}$5 a new provision that allows the bond holder to demand early redemption at a predetermined date prior to the maturity.

\textbf{Assumption 6}: At a predetermined date , the bond holder can receive $E$(unit of currency)(the \textit{early redemption premium}) and end the bond contract.

By this provision, the bond holder has the credit bond defined by assumptions 1$\mathrm{\sim}$5 together with the right to sell this bond for the exercise price $E$ at time $T_{1} $. Thus the Assumption 6 provides the bond holder the bond (put) option with maturity $T_{1} $. This is a bond option on zero coupon bond with credit risk. If the default event occurs, then the bond holder receives the default recovery and the bond contract is ended, and thus the bond option contract is also ended. Therefore this bond option contract is a kind of barrier option contract and its price is the solution to the following problem
\begin{equation} \label{eq19} 
\frac{\partial P}{\partial t} +\frac{1}{2} \sigma ^{2} V^{2} \frac{\partial ^{2} P}{\partial V^{2} } +(r-q)V\frac{\partial P}{\partial V} -rP=0,\; \; V>V_{b} (t),\; \; 0\le t<T_{1}  
\end{equation} 
\begin{equation} \label{eq20} 
P(V_{b} (t),\; t)=0,\; \; \; \; \; \; \; \; \; \; \; \; \; \; \; \; \; \; 0\le t<T_{1}  
\end{equation} 
\[P(V,\; T_{1} )=(E-B(V,\; T_{1} ))^{+} ,\; \; V>V_{b} (T_{1} )=be^{-a(T-T_{1} )} .\] 
Here $B(V,\; T_{1} )$ is provided by \eqref{eq16}. From Theorem 3, $B(V,\; T_{1} )$ is a strict increasing function on $V$ and we have
\[\inf _{V} B(V,\; \; T_{1} )=e^{-r(T-T_{1} )} R\; ,\; \; \sup _{V} B(V,\; \; T_{1} )=e^{-r(T-T_{1} )} .\] 
Therefore if $e^{-r(T-T_{1} )} R<E<e^{-r(T-T_{1} )} $, then there exists a unique number (called the \textit{early redemption boundary})  $K\ge V_{b} (T_{1} )$ $=be^{-a(T-T_{1} )} $ such that

$\; \; \; \;\; \; \; \; \; \; \;\; \; \; \;\; \; \; \;\; \; \; \;\;\; \; \; \;\; \; \; \;\; \; \; \;\; \; \; \; B(V,\; T_{1} )<E$ if $V<K$;\; \; \; \; $B(V,\; T_{1} )\ge E$ if $V\ge K$.

\noindent Thus the terminal value condition is written by
\[P(V,\; T_{1} )=E\cdot 1\{ V<K\} -B(V,\; T_{1} )\cdot 1\{ V<K\} ,\; \; V>V_{b} (T_{1} )=be^{-a(T-T_{1} )} . \] 
From \eqref{eq16}$\mathrm{\sim}$\eqref{eq18}, the terminal value condition can be written as
\begin{equation} \label{eq21} 
\begin{array}{l} {P(V,\; T_{1} )=[E-R\cdot e^{-r(T-T_{1} )} ]\cdot 1\{ V<K\} -(1-R)e^{-r(T-T_{1} )} N(d_{1} )\cdot 1\{ V<K\} } \\ {\; \; \; \; \; \; \; \; \; \; \; \; \; \; +(1-R)e^{-r(T-T_{1} )} \left(\frac{V}{be^{-a(T-T_{1} )} } \right)^{1-\frac{2(r-q-a)}{\sigma ^{2} } } N(d_{2} )\cdot 1\{ V<K\} ,\; \; V>be^{-a(T-T_{1} )} } \end{array} 
\end{equation} 

The problem \eqref{eq19}, \eqref{eq20}, \eqref{eq21} is a pricing problem for barrier option with moving barrier. Using new variables , then this problem is changed to the following problem. 
\begin{equation} \label{eq22}
\frac{\partial p}{\partial t} +\frac{1}{2} \sigma ^{2} x^{2} \frac{\partial ^{2} p}{\partial x^{2} } +(r-q-a)x\frac{\partial p}{\partial x} -rp=0,\; \; b<x<+\infty ,\; \; 0<t<T_{1}  
\end{equation} 
\begin{equation} \label{eq23} 
p(b,\; \; t)=0,\; \; 0\le t<T_{1}  
\end{equation} 
\begin{equation} \label{eq24} 
\begin{array}{l} {p(x,\; T_{1} )=[E-R\cdot e^{-r(T-T_{1} )} ]\cdot 1\{ b<x<Ke^{a(T-T_{1} )} \} } \\ {\; \; \; \; \; \; \; \; \; \; \; \; -(1-R)e^{-r(T-T_{1} )} N(d_{1} )\cdot 1\{ b<x<Ke^{a(T-T_{1} )} \} } \\ {\; \; \; \; \; \; \; \; \; \; \; \; +(1-R)e^{-r(T-T_{1} )} (x/b)^{1-\frac{2(r-q-a)}{\sigma ^{2} } } N(d_{2} )\cdot 1\{ b<x<Ke^{a(T-T_{1} )} \} =f_{1} +f_{2} +f_{3} .} \end{array} 
\end{equation}
Here 
\[d_{1} =\frac{\ln x/b+(r-q-a-\sigma ^{2} /2)(T-T_{1} )}{\sigma \sqrt{T-T_{1} } } =\delta (x/b,\; \; T-T_{1} ,\; \; 0,\; \; T-T_{1} ),\] \\
\[d_{2} =\frac{\ln b/x+(r-q-a-\sigma ^{2} /2)(T-T_{1} )}{\sigma \sqrt{T-T_{1} } } =\delta (b/x,\; \; T-T_{1} ,\; \; 0,\; \; T-T_{1} ),\]
\[f_{1} (x)=[E-R\cdot e^{-r(T-T_{1} )} ]\cdot 1\left\{ b<x<Ke^{a(T-T_{1} )} \right\} ,\; \] 
\[f_{2} (x)=-(1-R)e^{-r(T-T_{1} )} N(d_{1} )\cdot 1\left\{ b<x<Ke^{a(T-T_{1} )} \right\} ,\] 
\[f_{3} (x)=(1-R)e^{-r(T-T_{1} )} \left(\frac{x}{b} \right)^{1-\frac{2(r-q-a)}{\sigma ^{2} } } N(d_{2} )\cdot 1\left\{ b<x<Ke^{a(T-T_{1} )} \right\} .\] 

The problem \eqref{eq22}$\mathrm{\sim}$\eqref{eq24} is a pricing problem of barrier option with interest rate $r$, dividend rate $q+a$, constant barrier $b$ and can be solved by image solution method(\cite{Bu2001}). That is, if the solution to the problem \eqref{eq22} and \eqref{eq24} is $u(x,\; \; t)$, then the solution to the problem \eqref{eq22}, \eqref{eq23} and \eqref{eq24} is provided by
\begin{equation} \label{eq25} 
p(x,\; \; t)=u(x,\; \; t)-(x/b)^{1-2(r-q-a)/\sigma ^{2} } u(b^{2} /x,\; \; t).               
\end{equation} 

Denote the solution to \eqref{eq22} with the terminal value $p(x,\; T_{1} )=f_{i} (x)$ by $u_{i} (x,\; \; t)$. 

Since $f_{1} (x)=[E-R\cdot e^{-r(T-T_{1} )} ]\cdot (1\{ x>b\} - 1\{ x>Ke^{a(T-T_{1} )} \})$, applying Theorem 1 when $\beta =0,\; \; i=0,\; \; L=0,\; \; r=r,\;\;$$q=q+a$, $K=Ke^{a(T-T_{1} )} $ and when $K=b$, and considering $N_{2} (+\infty ,\; a_{2} ;\; A)=N(a_{2} )$, then we have
\[u_{1} (x,\; \; t)=[E-R\cdot e^{-r(T-T_{1} )} ]\cdot e^{-r(T_{1} -t)} [N(b_{1})-b_{2})]=[E\cdot e^{-r(T_{1} -t)} -R\cdot e^{-r(T-t)} ][N(b_{1} )-N(b_{2} )].\] 
Here
\[b_{1} =\delta (\frac{x}{b} ,\; \; T_{1} -t,\; \; 0,\; \; T_{1} -t)=\frac{\ln \frac{x}{b} +(r-q-a-\sigma ^{2} /2)(T_{1} -t)}{\sigma \sqrt{T_{1} -t} } .\] 
\[b_{2} =\delta (\frac{x}{Ke^{a(T-T_{1} )} } ,\; \; T_{1} -t,\; \; 0,\; \; T_{1} -t)=\frac{\ln \frac{x}{Ke^{a(T-T_{1} )} } +(r-q-a-\sigma ^{2} /2)(T_{1} -t)}{\sigma \sqrt{T_{1} -t} } ,\]

Since $f_{2} (x)=-(1-R)e^{-r(T-T_{1} )} N(d_{1} )\cdot (1\{ x>b\} - 1\{ x>Ke^{a(T-T_{1} )} \})$, applying Theorem 1 when $\beta =0,\; \; i=1,\; \; L=b,\;\;$$r=r,\; \; q=q+a$,$\;\;K=Ke^{a(T-T_{1} )} $ and when $K=b$, then we have
\[u_{2} (x,\; \; t)=-(1-R)e^{-r(T-t)} [N_{2} (\overline{d}_{1} ,\; \; b_{1} ;\rho)-N_{2} (\overline{d}_{1} ,\; \; b_{2} ;\rho)].\] 
Here
\[\overline{d}_{1} =\delta (\frac{x}{b} ,\; \; T-t,\; \; 0,\; \; T-t)=\frac{\ln \frac{x}{b} +(r-q-a-\frac{\sigma ^{2} }{2} )(T-t)}{\sigma \sqrt{T-t} }   \] 
\[\rho=\sqrt{(T_{1} -t)/(T-t)}\]

Since $f_{3} (x)=(1-R)e^{-r(T-T_{1} )} \left(\frac{x}{b} \right)^{1-\frac{2(r-q-a)}{\sigma ^{2} } } N(d_{2} )\cdot (1\{ x<Ke^{a(T-T_{1} )} \} -1\{ x<b\} )$, so applying Theorem 1 when $\beta =1-2\cdot (r-q-a)/\sigma ^{2} ,$$\;\;i=-1,\; \; L=b^{-1} ,\; \; r=r,$$\;\;q=q+a$,$\;\;K=Ke^{a(T-T_{1} )} $ and when $K=b$, and considering $\mu (\beta )=-r$, then we have
\[\begin{array}{l} {u_{3} (x,\; \; t)=(1-R)e^{-r(T-T_{1} )} e^{\mu (\beta )(T_{1} -t)} (x/b)^{\beta } [N_{2} (\overline{d}_{2} ,\; \; b_{3} ;-\rho)-N_{2} (\overline{d}_{2} ,\; \; b_{4} ;-\rho)]=} \\ {\; \; \; \; \; \; \; \; \; \; \; =(1-R)e^{-r(T-t)} (x/b)^{1-2(r-q-a)/\sigma ^{2} } [N_{2} (\overline{d}_{2} ,\; \; b_{3} ;-\rho)-N_{2} (\overline{d}_{2} ,\; \; b_{4} ;-\rho)]} \end{array}.\] 
Here
\[b_{3} =\delta \left(\frac{x}{b} ,\; T_{1} -t,\; \beta (T_{1} -t),\; T_{1} -t\right)=\frac{\ln \frac{x}{b} -(r-q-a-\frac{\sigma ^{2} }{2} )(T_{1} -t)}{\sigma \sqrt{T_{1} -t} } \] 
\[b_{4} =\delta \left(\frac{x}{Ke^{a(T-T_{1} )} } ,\; T_{1} -t,\; \beta (T_{1} -t),\; T_{1} -t\right)=\frac{\ln \frac{x}{Ke^{a(T-T_{1} )} } -(r-q-a-\frac{\sigma ^{2} }{2} )(T_{1} -t)}{\sigma \sqrt{T_{1} -t} } \] 
\[\overline{d}_{2} =\delta \left(\frac{b}{x} ,\; T-2T_{1} +t,\; -\beta (T_{1} -t),\; T-t\right)=\frac{\ln \frac{b}{x} +(r-q-a-\frac{\sigma ^{2} }{2} )(T-t)}{\sigma \sqrt{T-t} }.\] 

Therefore the solution to the problem \eqref{eq22} and \eqref{eq24} is given as follows:
\[\begin{array}{l} {u(x,\; \; t)=[E\cdot e^{-r(T_{1} -t)} -R\cdot e^{-r(T-t)} ][N(b_{1} (x))-N(b_{2} (x))]} \\ {\; \; \; \; \; \; \; \; \;\; -(1-R)e^{-r(T-t)} [N_{2}(\overline{d}_1(x), b_{1}(x) ;\rho)-N_{2} (\overline{d}_1(x), b_{2}(x) ;\rho)]} \\ {\;\;\; \; \; \; \; \; \; \; +(1-R)e^{-r(T-t)} (x/b)^{1-2(r-q-a)/\sigma ^{2} } [N_{2} (\overline{d}_2(x), b_3(x) ;-\rho)-N_{2} (\overline{d}_2(x), b_4(x); \;-\rho)]} \end{array}.\] 
And the solution to the problem \eqref{eq22}, \eqref{eq23} and \eqref{eq24} is given by \eqref{eq25}, so we have
\[\begin{array}{l} {p(x,\; \; t)=[Ee^{-r(T_{1} -t)} -R\cdot e^{-r(T-t)} ]\{ N(b_1(x))-N(b_2(x))} \\ {\; \; \; \; \; \; \; \; \; \; \; \; \; \; \; \; \; \; \; \; \; \; \; \; \; \; \; \; \; \; \; \; \; \; \; \; \; \; \; \; \; \; \; \; \; \; \; -(x/b)^{1-2(r-q-a)/\sigma ^{2} } [N(b_{1} (b^{2}/x))-N(b_{2} (b^{2} /x))\} } \\ {\; \; -(1-R)e^{-r(T-t)} \{ N_{2}(\overline{d}_1(x), b_{1}(x) ;\rho)-N_{2}(\overline{d}_1(x), b_{2}(x) ;\rho)-} \\ {\; \; -(x/b)^{1-2(r-q-a)/\sigma ^{2} } [N_{2}(\overline{d}_1(b^{2} /x),\; \; b_{1} (b^{2} /x);\; \rho)-N_{2} (\overline{d}_{1} (b^{2} /x),\; \; b_{2} (b^{2} /x);\; \rho)]\} } \\ {\; \; +(1-R)e^{-r(T-t)} \{ (x/b)^{1-2(r-q-a)/\sigma ^{2} } [N_{2}(\overline{d}_2(x), b_{3}(x) ;-\rho)-N_{2}(\overline{d}_2(x), b_{4}(x);\; -\rho)-} \\ {\; \; \; \; \; \; \; \; \; \; \; \; \; \; \; \; \; \; \; \; \; \; \; \; -[N_{2} (\overline{d}_{2}(b^{2} /x),\; b_{3}(b^{2} /x)\; ; -\rho)-N_{2} (\overline{d}_{2}(b^{2} /x),\; \; b_{4} (b^{2} /x);-\rho)]\} .} \end{array}\] 
Here if we consider the relationships
\[x=Ve^{a(T-t)},\;\;\frac{x}{b}=\frac{V}{be^{-a(T-t)}},\;\;\frac{x}{Ke^{a(T-T_{1})}}=\frac{Ve^{a(T_{1}-t)}}{K}, \]
\noindent then we have
\[b_{1} (x)=\frac{\ln \frac{V}{be^{-a(T-t)} } +(r-q-a-\sigma ^{2} /2)(T_{1} -t)}{\sigma \sqrt{T_{1} -t} } :=b_{1} \] 
\[b_{2} (x)=\frac{\ln \frac{V}{K} +(r-q-\frac{\sigma ^{2} }{2} )(T_{1} -t)}{\sigma \sqrt{T_{1} -t} } :=b_{2} \] 
\[b_{1} (b^{2} /x)=\frac{\ln \frac{be^{-a(T-t)} }{V} +(r-q-a-\sigma ^{2} /2)(T_{1} -t)}{\sigma \sqrt{T_{1} -t} } :=\tilde{b}_{1} \] 
\[b_{2} (b^{2} /x)=\frac{\ln \frac{be^{-a(T-T_{1} )} }{K} \cdot \frac{be^{-a(T-t)} }{V} +(r-q-\sigma ^{2} /2)(T_{1} -t)}{\sigma \sqrt{T_{1} -t} } :=\tilde{b}_{2} \] 
\[\overline{d}_{1}(x)=\frac{\ln \frac{V}{be^{-a(T-t)} } +(r-q-a-\frac{\sigma ^{2} }{2} )(T-t)}{\sigma \sqrt{T-t} }(=d_{1} \;\text{in}\; \eqref{eq18}):=d\]
\[\overline{d}_{1}(b^2/x)=\frac{\ln \frac{be^{-a(T-t)} }{V} +(r-q-a-\frac{\sigma ^{2} }{2} )(T-t)}{\sigma \sqrt{T-t} }(=d_{2} \; \text{in}\; \eqref{eq18}):=\tilde{d}\]
\[\overline{d}_{2}(x)=\frac{\ln \frac{be^{-a(T-t)} }{V} +(r-q-a-\frac{\sigma ^{2} }{2} )(T-t)}{\sigma \sqrt{T-t} }(=d_{2} \; \text{in}\; \eqref{eq18}=\tilde{d})\]
\[b_{3} (x)=\frac{\ln \frac{V}{be^{-a(T-t)} } -(r-q-a-\frac{\sigma ^{2} }{2} )(T_{1} -t)}{\sigma \sqrt{T_{1} -t} } (=-\tilde{b}_{1})\]
\[b_{4}(x)=\frac{\ln \frac{Ve^{a(T -t)} }{Ke^{a(T-T_1)}} -(r-q-a-\frac{\sigma ^{2} }{2} )(T_{1} -t)}{\sigma \sqrt{T_{1} -t} } :=-\tilde{b}_5\]
\[b_{3} (b^{2} /x)=\frac{\ln \frac{be^{-a(T-t)} }{V} -(r-q-a-\frac{\sigma ^{2} }{2} )(T_{1} -t)}{\sigma \sqrt{T_{1} -t} } =-b_{1} \] 
\[b_{4} (b^{2} /x)=\frac{\ln \frac{be^{-a(T-T_{1} )} }{K} \cdot \frac{be^{-a(T-t)} }{V} -(r-q-a-\frac{\sigma ^{2} }{2} )(T_{1} -t)}{\sigma \sqrt{T_{1} -t} } :=-b_{5} \] 
\[\overline{d}_{2} (b^{2} /x)=\frac{\ln \frac{V}{be^{-a(T-t)} } +(r-q-a-\frac{\sigma ^{2} }{2} )(T-t)}{\sigma \sqrt{T-t} }=d_{1}\; \text{in}\; \eqref{eq18}=d\] 
and 
\begin{align}
P(V,\; t)&=p(Ve^{a(T-t)} ,\; \; t) \nonumber \\ 
&=[E\cdot e^{-r(T_{1} -t)} -R\cdot e^{-r(T-t)} ]\{ N(b_{1} )-N(b_{2} )-\left(\frac{V}{be^{-a(T-t)} } \right)^{1-2(r-q-a)/\sigma ^{2} } [N(\tilde{b}_{1})-N(\tilde{b}_{2})]\} \nonumber \\
&-(1-R)e^{-r(T-t)} \Big\{N_{2} (d ,\; b_{1} ;\;\rho)-N_{2} (d ,\; b_{2}; \;\rho)-\left(\frac{V}{be^{-a(T-t)} } \right)^{1-2(r-q-a)/\sigma ^{2} } \cdot \nonumber \\
& \; \; \; \; \; \; \;\; \; \; \; \; \; \; \; \; \;\cdot [N_{2} (\tilde{d} ,\; \tilde{b}_{1} ;\;\rho)-N_{2} (\tilde{d} ,\; \tilde{b}_{2} ;\;\rho)]\; +[N_{2} (d ,\;  -{b}_{1} \; ;\; -\rho)-N_{2} (d ,\; \; -{b}_{5} \; ;\; -\rho)] \nonumber \\ 
&\; \; \; \; \; \; \;\; \; \; \; \; \; \; \; \; \;-\left(\frac{V}{be^{-a(T-t)} } \right)^{1-2(r-q-a)/\sigma ^{2} } [N_{2} (\tilde{d} ,\; -\tilde{b}_{1} ;\;-\rho)-N_{2} (\tilde{d} ,\; -\tilde{b}_{5} ;\;-\rho)]\Big\}.\label{eq26}  
\end{align} 

Thus we have proved the following theorem.

\begin{theorem} Assume that $e^{-r(T-T_{1} )} R<E<e^{-r(T-T_{1} )} $. The price for bond put option (the solution to the problem \eqref{eq19}, \eqref{eq20}, \eqref{eq21}) is given by \eqref{eq26}.
\end{theorem}

\begin{remark} Assumption 6 gives the holder of the bond with expiry date $T$ the right to receive early repayment of debts at the date $T_{1} $ prior to the expiry date. The repayment $E$ must satisfy $e^{-r(T-T_{1} )} R<E<e^{-r(T-T_{1} )} $. \eqref{eq26} is just the early redemption premium which the holder should pay more due to this early redemption provision. Thus the price in the time interval $[0,\; \; T_{1} ]$ of the \textit{puttable bond} (the \textit{bond with early redemption provision}) is the sum of $B$(given by \eqref{eq16}) and $P$(given by \eqref{eq26}), and the price in the time interval $(T_{1} ,\; \; T]$ is equal to $B$.
\end{remark}

Next, we assume that the bond issuing company added to the bond defined by the Assumptions 1$\mathrm{\sim}$5 a new provision that allows the company to pay back the debt early at a predetermined date prior to the maturity.

\textbf{Assumption 7}: At a predetermined date , the bond issuing company can give $E$(unit of currency) and call the bond.

By this provision, the bond issuing company has the right to purchase this bond defined by the Assumptions 1$\mathrm{\sim}$5 for the exercise price $E$ at time $T_{1} $. Thus the Assumption 7 provides the company the bond (call) option with maturity $T_{1} $. This is a bond option on zero coupon bond with credit risk. If the default event occurs, then the bond holder receives the default recovery and the bond contract is ended, and thus the bond option contract is ended, too. Therefore this bond option contract is a kind of barrier option contract and its price $P=C(V,\; \; t)(0\le t\le T_{1} )$ satisfies \eqref{eq19} and \eqref{eq20}, and we have the following terminal value condition
\[P_{T_{1} } =C(V,\; T_{1} )=(B(V,\; T_{1} )-E)^{+} ,\; \; V>V_{b} (T_{1} )=be^{-a(T-T_{1} )} .\] 
Here $B(V,\; T_{1} )$is given by \eqref{eq16} and thus there exists a unique number $K\ge V_{b} (T_{1} )=be^{-a(T-T_{1} )} $such that 

\[B(V,\; T_{1} )<E \;\;\text{if}\; V<K,\;\;\; B(V,\; T_{1} )\ge E \;\;\text{if} \;V\ge K.\]

\noindent And the terminal value condition of the option is written by 
\[C(V,\; T_{1} )=B(V,\; T_{1} )\cdot 1\{ V\ge K\} -E\cdot 1\{ V\ge K\} . \] 
Thus from \eqref{eq16}$\mathrm{\sim}$\eqref{eq18}, we have
\begin{equation} \label{eq27} 
\begin{array}{l} {C(V,\; T_{1} )=[R\cdot e^{-r(T-T_{1} )} -E]\cdot 1\{ V\ge K\} +(1-R)e^{-r(T-T_{1} )} N(d_{1} )\cdot 1\{ V\ge K\} } \\ {\; \; \; \; \; \; \; \; \; \; \; \; \; \; -(1-R)e^{-r(T-T_{1} )} \left(\frac{V}{be^{-a(T-T_{1} )} } \right)^{1-\frac{2(r-b-a)}{\sigma ^{2} } } N(d_{2} )\cdot 1\{ V\ge K\} } \end{array} 
\end{equation} 

Using the same method of Theorem 4, we can obtain the following theorem.

\begin{theorem} If $e^{-r(T-T_{1} )} R<E<e^{-r(T-T_{1} )} $, then the price of the bond (call) option (the solution to the problem \eqref{eq19}, \eqref{eq20} and \eqref{eq27}) is represented as follows:
\begin{align}
P(V,\; t)&=p(Ve^{a(T-t)} ,\; \; t) \nonumber \\
&=[R\cdot e^{-r(T-t)} -E\cdot e^{-r(T_{1} -t)} ]\left[N(b_{2} )-\left(\frac{V}{be^{-a(T-t)} } \right)^{1-2(r-q-a)/\sigma ^{2} } N(\tilde{b}_{2} )\right]+\nonumber \\ 
&+(1-R)e^{-r(T-t)} \Big[N_{2} (d ,\; \; b_{2} ;\;\rho)-\left(\frac{V}{be^{-a(T-t)} } \right)^{1-2(r-q-a)/\sigma ^{2} } N_{2} (\tilde{d} ,\; \; \tilde{b}_{2};\;\rho)  \nonumber \\
&+N_{2} (d ,\; \; -{b}_{5};\;-\rho)-\left(\frac{V}{be^{-a(T-t)} } \right)^{1-2(r-q-a)/\sigma ^{2} } N_{2} (\tilde{d} ,\; -\tilde{b}_{5} ;\;-\rho)\Big].\label{eq28} 
\end{align} 
\end{theorem}

\begin{remark} Assumption 7 gives the bond issuing company(debtor) the right to call the bond from the bond holder at the date $T_{1} $ prior to expiry date with the early repayment $E$ and obliged the bond holder. \eqref{eq28} is just the premium that the company should pay to the holder due to this early repayment provision. So the price in the time interval $[0,\; \; T_{1} ]$ of this callable bond is equal to the difference of $P$(given by \eqref{eq26}) from $B$(given by \eqref{eq16}), and the price in the time interval $(T_{1} ,\; \; T]$ is equal to $B$.
\end{remark}

\textbf{Numerical Experiments}

When $r=0.04,\; \; q=0,\; \; \sigma =0.5\; ;\; \; b=100,\; \; a=0;\; \; R=0.7,\; \; T=2,\; \; T_{1} =1,\; \; E=0.9$, Figure 1 is the plot (t, B) of the price given by \eqref{eq16} of ordinary credit bond without early redemption provision. Figure 2 shows the plot (V, B) for firm value V at time T${}_{1}$ and 0 and the early redemption boundary. Figure 3 is the time plot (t, P) of the bond option price (early redemption premium) given by \eqref{eq26}, and Figure 4 is the graph (V, P) for firm value V at some fixed time. Figure 5 is the time plot of the price of the puttable bond (with early redemption provision).

\begin{center}
\noindent \includegraphics*[width=3.5in, height=2.5in, keepaspectratio=false]{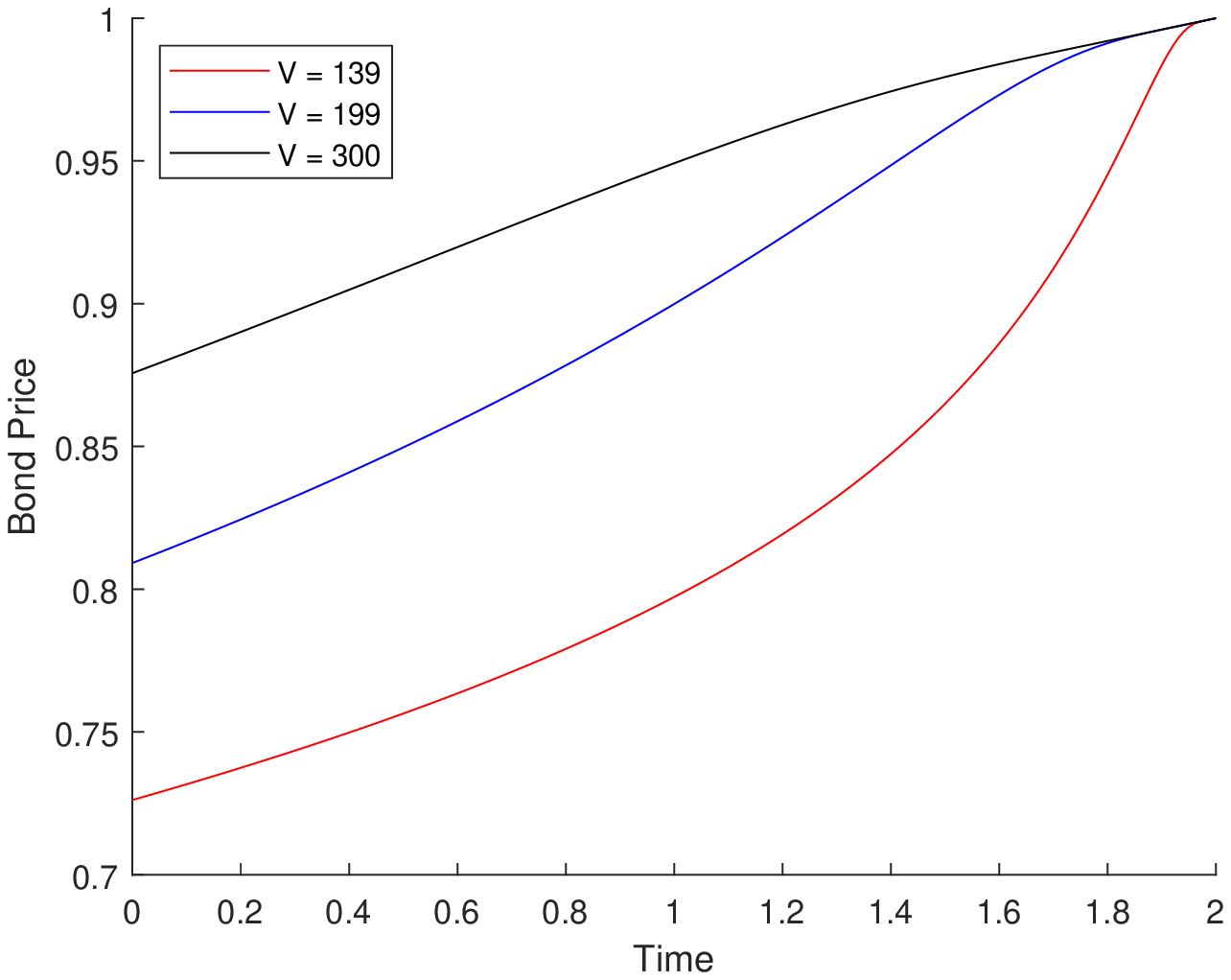}    

\noindent Figure 1. Plot (t, B) in [0, T] when V=\textbf{139}, \textbf{199}, \textbf{300;}
\end{center}
 
\begin{center}
\noindent \includegraphics*[width=3.5in, height=2.5in, keepaspectratio=false]{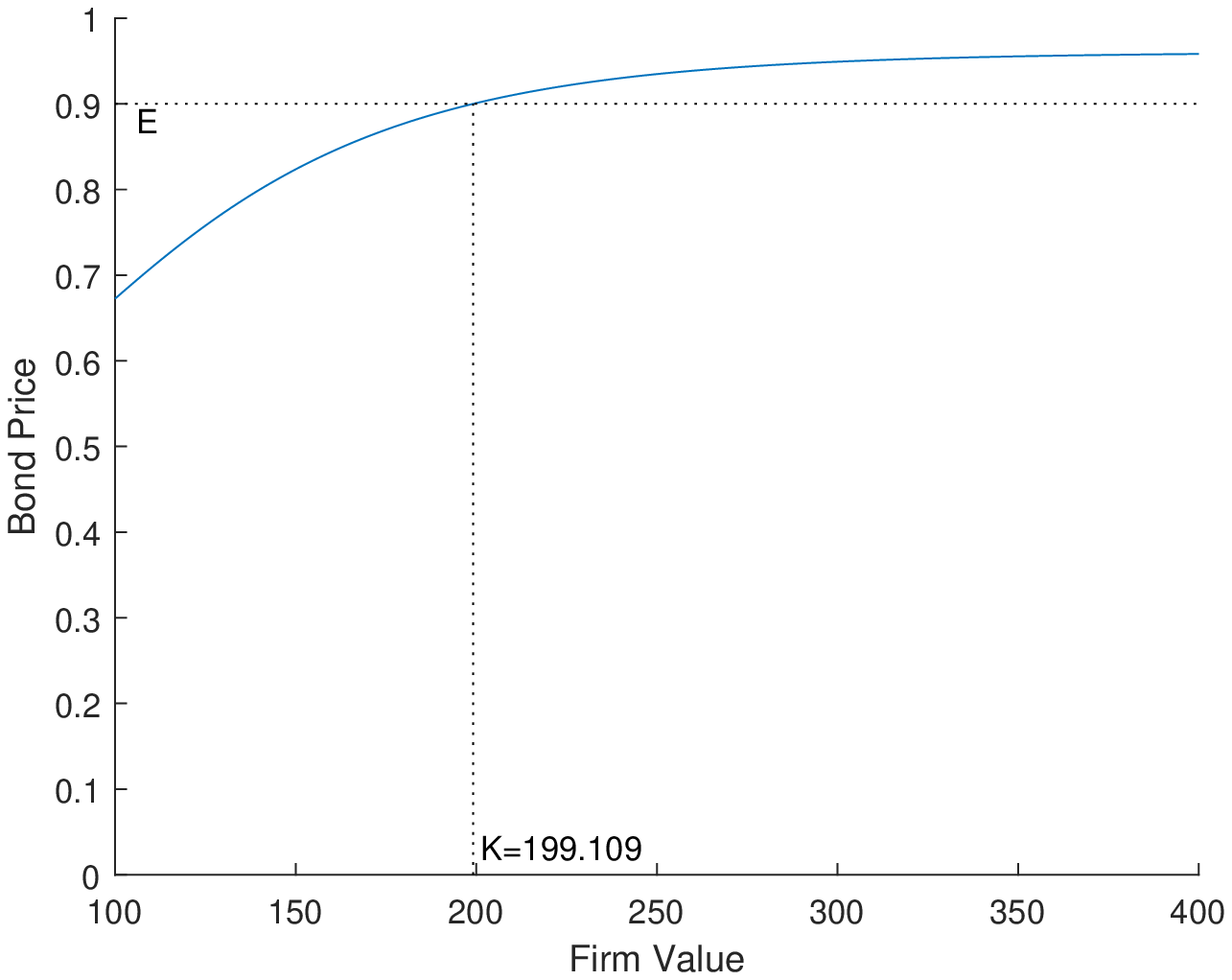}

\noindent Figure 2. Plot (V, B) when t =\textbf{T${}_{1}$${}_{\ }$}, early redemption boundary K=199.109
\end{center}

\begin{center}
\noindent \includegraphics*[width=3.5in, height=2.5in, keepaspectratio=false]{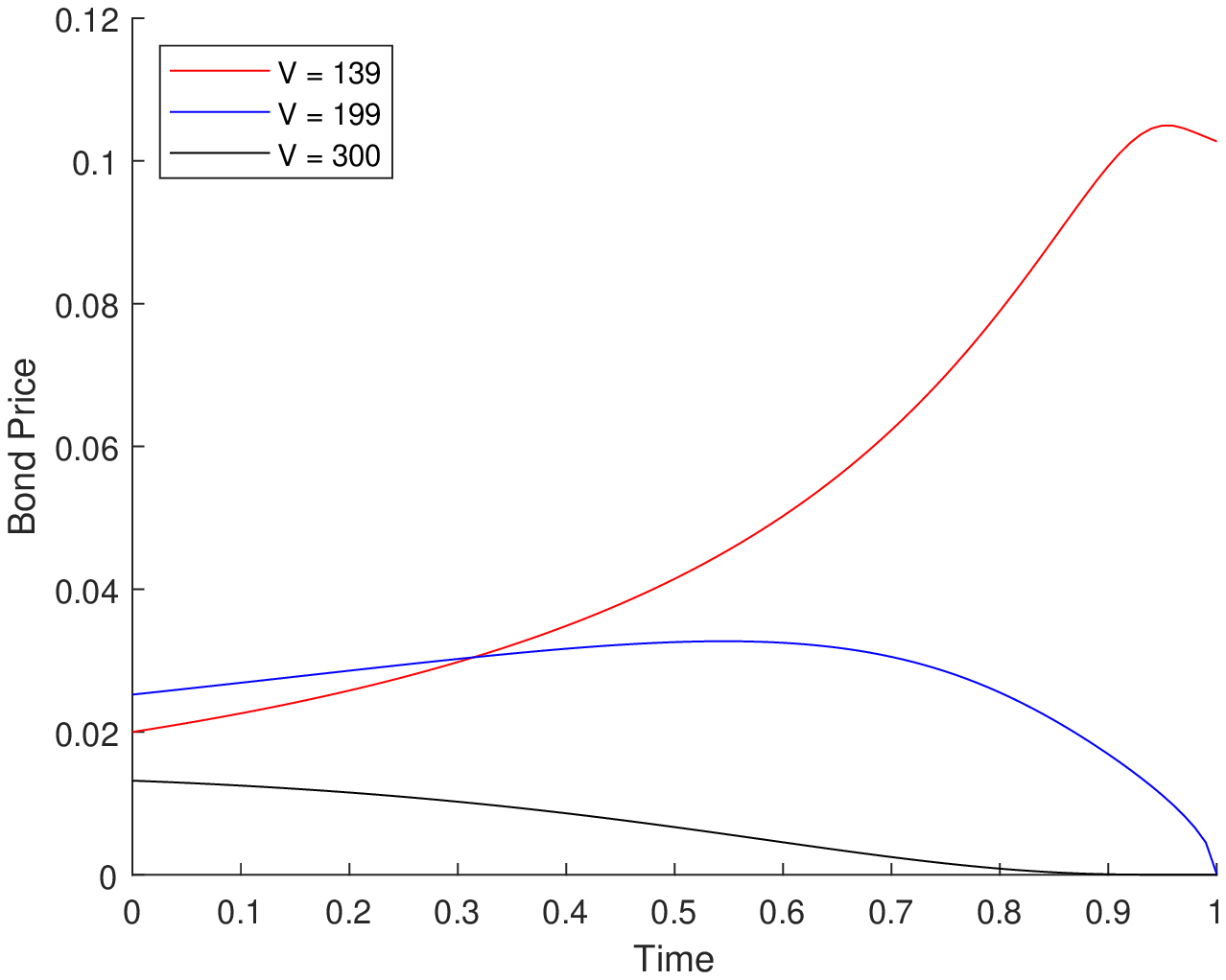}

\noindent Figure 3.  Plot (t, P) in [0, T${}_{1}$] when V=\textbf{139}, \textbf{199}, \textbf{300}
\end{center}

\begin{center}
\noindent \includegraphics*[width=3.5in, height=2.5in, keepaspectratio=false]{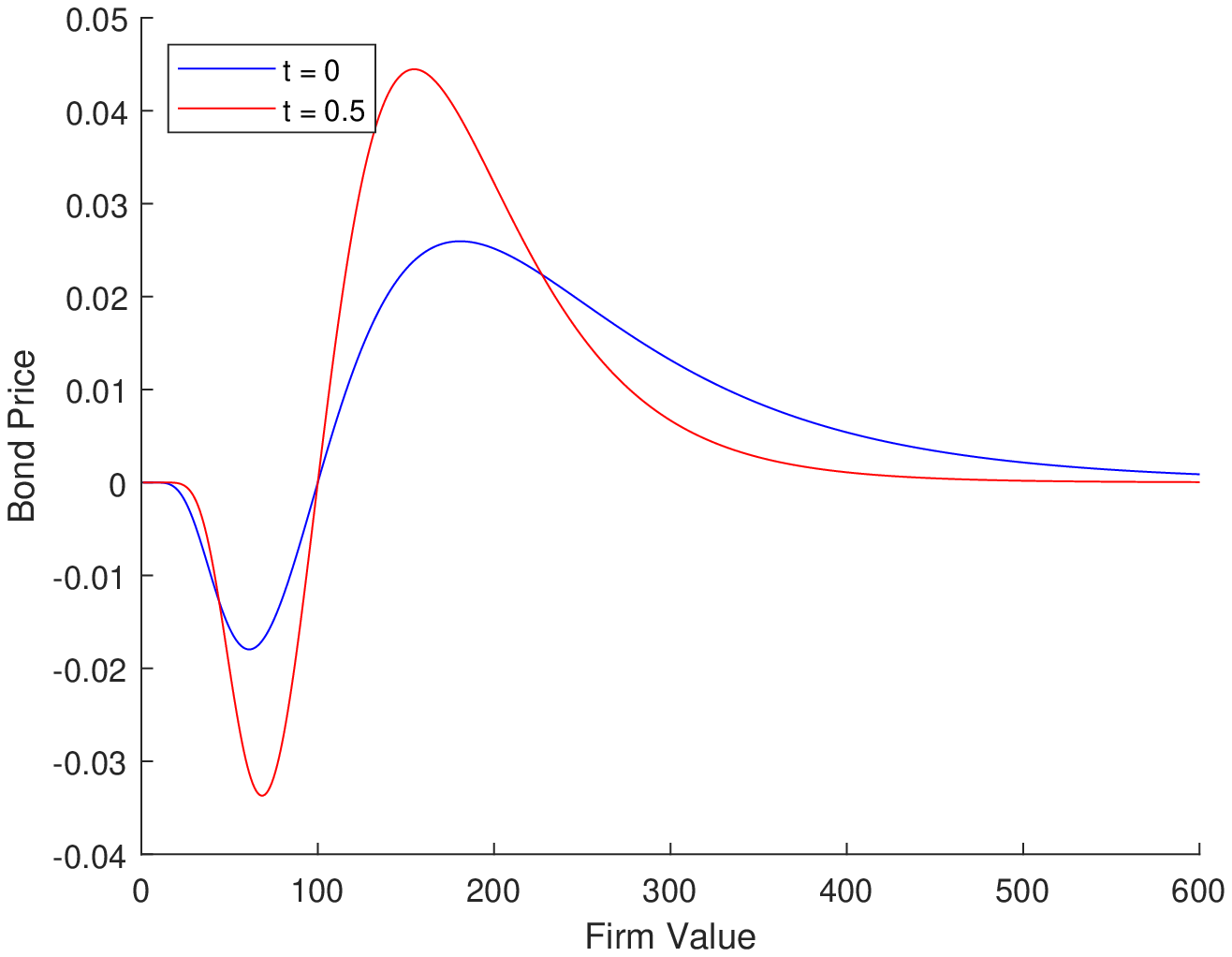}

\noindent Figure 4. Plot (V, P) when t = \textbf{0}, \textbf{T${}_{1}$/2}
\end{center}

\noindent 
\begin{center}
\includegraphics*[width=3.5in, height=2.5in, keepaspectratio=false]{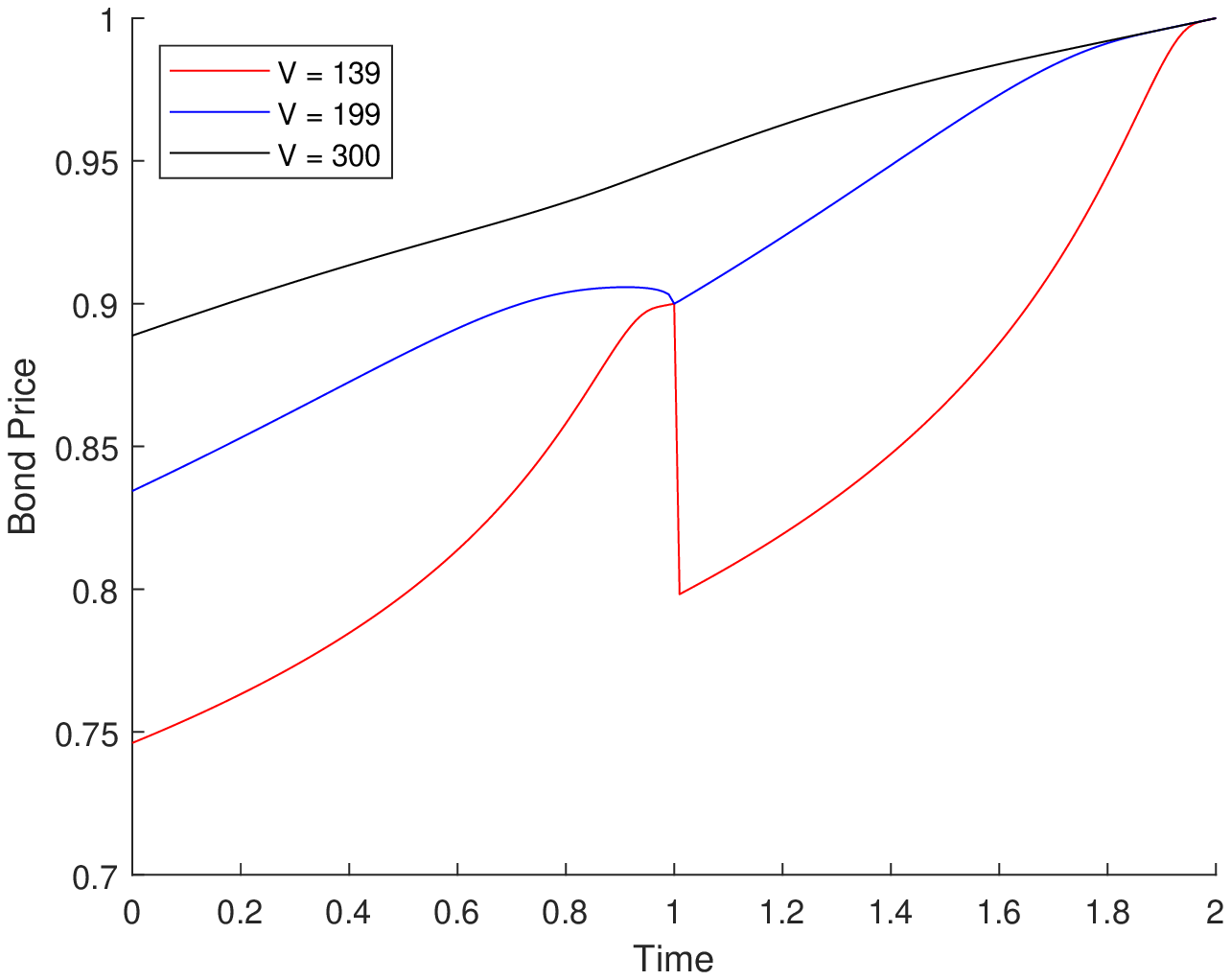}

\noindent Figure 5. Time Plot of the price of the puttable bond, (t, P+B) in [0, T${}_{1}$], (t, B) in [T${}_{1}$, T] when V=\textbf{139}, \textbf{199}, \textbf{300}
\end{center}

\noindent 

The results of the experiment show that the theoretical treatment in the above is correct. Figure 1 and Fig. 2 show that increasing property for firm value variable of the price function of ordinary credit bond without early redemption provision and uniqueness of early redemption boundary. Fig. 3 and Fig. 4 show the early redemption premium which the holder of puttable bond should pay more due to this early redemption provision. Fig. 5 shows that the price of our puttable bond increases with respect to the firm value variable.

\end{document}